\documentclass[11pt,a4paper]{article}

\usepackage{amsmath,amssymb,mathtools}
\usepackage{amsthm}
\usepackage{booktabs}
\usepackage{array}
\usepackage[a4paper,margin=1in]{geometry}
\usepackage[numbers,sort&compress]{natbib}
\usepackage{microtype}
\usepackage{xcolor}
\usepackage{tikz}
\usepackage{graphicx}
\usepackage[hidelinks]{hyperref}

\graphicspath{{SIDMA_review_corrected_20260911/}}
\newtheorem{theorem}{Theorem}[section]
\newtheorem{proposition}{Proposition}[section]
\newtheorem{lemma}{Lemma}[section]
\newtheorem{corollary}{Corollary}[section]
\theoremstyle{definition}
\newtheorem{definition}{Definition}[section]
\theoremstyle{remark}
\newtheorem{remark}{Remark}[section]

\newenvironment{claim}{\par\medskip\noindent\textit{Claim.}\ }{\par\medskip}
\newcommand{\F}{\mathbb{F}}
\DeclareMathOperator{\supp}{supp}
\newcommand{\FullHostFigureWidth}{0.68\textwidth}

\title{A Four-Connected Graph without a Legal System}
\author{Qiuyu Chen\\[0.4em]
\small Department of Computer Science, Shanghai Jiao Tong University}
\date{}

\begin{document}

\maketitle

\def\JNWCONTENTONLY{1}
\ifdefined\JNWCONTENTONLY
% The standalone arXiv article supplies its own preamble and title block.
\else
\documentclass[review]{siamart251216}

\usepackage{amssymb}
\usepackage{xcolor}
\usepackage{tikz}
\usepackage{graphicx}

% SIAM's class loads hyperref; keep links unobtrusive for review.
\hypersetup{colorlinks=true,linkcolor=black,citecolor=black,urlcolor=black}
\emergencystretch=1em

\newcommand{\F}{\mathbb{F}}
\DeclareMathOperator{\supp}{supp}

% The SIAM class provides theorem, lemma, proposition, corollary, and
% definition environments. Remarks are declared through SIAM's remark hook.
\newsiamremark{remark}{Remark}
\newenvironment{claim}{\par\medskip\noindent\textit{Claim.}\ }{\par\medskip}

\title{A Four-Connected Graph without a Legal System}
\author{Qiuyu Chen\thanks{Department of Computer Science, Shanghai Jiao Tong University
(\email{canghaimeng@sjtu.edu.cn}).}}
\headers{A Four-Connected Graph without a Legal System}{Q. Chen}

\begin{document}
\maketitle
\fi

\begin{abstract}
In a 2021 paper, Jankiewicz, Norin, and Wise asked whether there exists a finite $4$-connected graph of girth at least four and nonnegative Charney--Davis curvature such that no $4$-connected ordinary subgraph admits a legal system. We construct such a graph by starting from the hexagonal prism and attaching three $K_{3,4}$-based caps along pairwise disjoint induced $4$-cycles. The key structural input is a restriction theorem showing that a legal system on an induced-$4$-cycle amalgam restricts to each side, so the obstruction carried by the negatively curved prism survives the attachments. The resulting $33$-vertex graph is $4$-regular and $4$-connected, has girth four and Charney--Davis curvature one, and, by $4$-regularity, is its own unique $4$-connected ordinary subgraph.
\end{abstract}

\ifdefined\JNWCONTENTONLY
\else

\fi

\section{Introduction}

In 2021, Jankiewicz, Norin, and Wise introduced legal systems of vertex-moves as a combinatorial device in the study of virtual algebraic fibrations of right-angled Coxeter groups~\cite{JNW21}; their construction is based on Bestvina--Brady Morse theory~\cite{BB97}. In that framework a legal system on a finite graph $G$ is a choice, at each vertex, of a subset containing the vertex and none of its neighbours, together with a starting subset, such that every set in the resulting affine orbit over $\F_2$ induces a nonempty connected subgraph and so does its complement~\cite[Section~2]{JNW21}. The same paper records a $2$-curvature invariant $\kappa_2(G):=1-|V(G)|/2+|E(G)|/4$, coinciding with the Charney--Davis curvature on triangle-free graphs, and proves that a legal system forces $\kappa_2(G)\ge 0$~\cite[Definition~3.1, Theorem~6.14]{JNW21}. They then ask whether there exists a finite $4$-connected graph $\Gamma$ of girth at least $4$ with $\kappa(\Gamma)\ge 0$ such that no $4$-connected ordinary subgraph admits a legal system. This is the $4$-connected question in~\cite[Problem~5.2]{JNW21}. Throughout, $4$-connected means that the graph has at least five vertices and remains connected after deletion of any set of at most three vertices, and an ordinary subgraph is not required to be induced; these are the conventions of~\cite[Definition~6.3]{JNW21}.

The significance of the Jankiewicz--Norin--Wise state-and-move criterion is that it turns finite combinatorial information into a certificate for virtual algebraic fibering. Subsequent work has used or generalized this mechanism for random right-angled Coxeter groups, higher finiteness properties of fibering kernels, and high-dimensional hyperbolic fibering constructions~\cite{FPGK18,SZ23,IMM24,LMSW25}. Thus legal systems form a reusable bridge between graph structure and geometric-group-theoretic fibering, making it natural to ask how strong the graph-theoretic obstruction to such systems can remain under favorable global conditions such as high connectivity and nonnegative curvature.

We answer this obstruction question by combining an obstruction-preserving graph operation with a regularity and connectivity construction. The structural input is an induced-$4$-cycle amalgam restriction: if an amalgam admits a legal system, then each side does. We start from the hexagonal prism, whose negative $2$-curvature rules out a legal system, and attach three seven-vertex caps with $K_{3,4}$ interiors along pairwise disjoint induced $4$-cycles. The restriction theorem preserves the obstruction through these attachments, while the resulting $4$-regular host is engineered to be $4$-connected with Charney--Davis curvature one. Its $4$-regularity also forces every $4$-connected ordinary subgraph to be the whole host. The complete $33$-vertex witness is shown in Figure~\ref{fig:full-host}, and its exact combinatorial construction is given in Definition~\ref{defn:G}.

\begin{figure}[!htbp]
\centering
\providecommand{\FullHostFigureWidth}{0.82\textwidth}
\includegraphics[width=\FullHostFigureWidth,trim=10bp 15bp 18bp 8bp,clip]{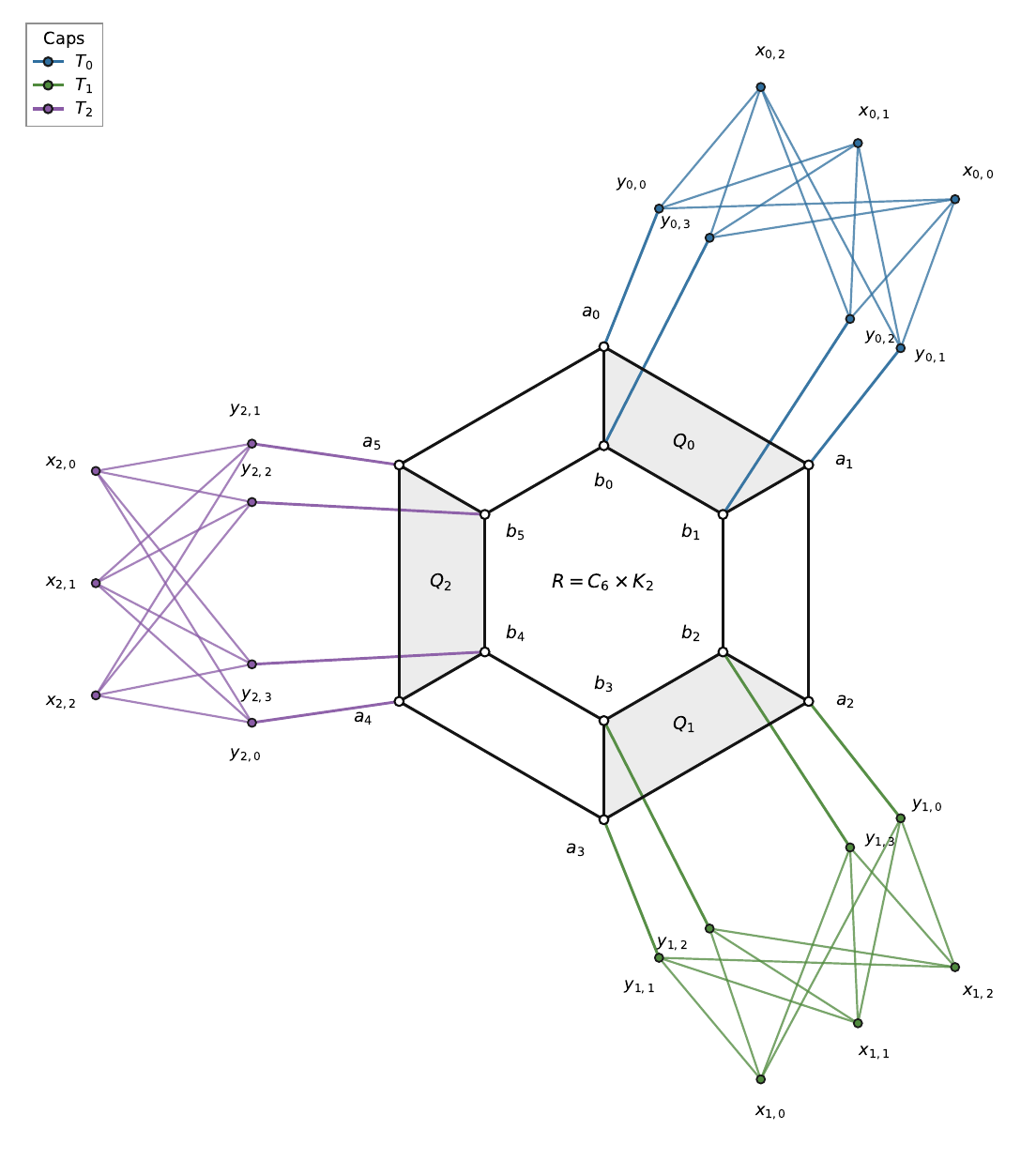}
\caption{The $33$-vertex witness $G$. Its central subgraph $R$ is the hexagonal prism. For each $t\in\{0,1,2\}$, the interior of the cap $T_t$ and its four matching edges are colored, while the shared interface cycle $Q_t$ is black. The drawing contains all $33$ vertices and all $66$ edges.}
\label{fig:full-host}
\end{figure}

\begin{theorem}[A four-connected graph without a legal system]
\label{thm:main}
There exists a finite simple undirected graph $G$ on $33$ vertices and $66$ edges such that
\begin{enumerate}
\item $G$ is $4$-regular, triangle-free, of girth four, and four-vertex-connected;
\item $\kappa_2(G)=\kappa(G)=1$;
\item no four-vertex-connected ordinary subgraph of $G$ admits a legal system.
\end{enumerate}
\end{theorem}

The proof is organized around two structural components. Section~\ref{sec:prelim} records the legal-system formalism, the curvature obstruction, and the affine characterization used later. Section~\ref{sec:amalgam} proves the induced-$4$-cycle restriction theorem, which preserves the absence of legal systems under the cap attachments. Section~\ref{sec:host} constructs the prism--cap graph and verifies the regularity, curvature, interface, and connectivity properties required of the host. Finally, Section~\ref{sec:proof} combines these ingredients with the elementary collapse of $4$-connected ordinary subgraphs in a connected $4$-regular graph to prove Theorem~\ref{thm:main}.

A Lean formalization of the main results has been carried out and is publicly available at \url{https://github.com/canghaimeng/JNWLegalSystemObstruction}.

\section{Preliminaries}
\label{sec:prelim}

The goal of this section is to fix the standing language of~\cite{JNW21} and to record the curvature obstruction together with the affine form of a legal system. Throughout, graphs are finite, simple, and undirected. Write $V(H)$ and $E(H)$ for the vertex-set and edge-set of a graph $H$, and write $H-U$ for the graph obtained by deleting a vertex set $U$ and all incident edges. A one-vertex graph is connected; the empty graph is not connected. The open neighbourhood of a vertex $v$ is written $N(v)$, and $N[v]:=N(v)\cup\{v\}$. Identify the power set of a finite set $V$ with the $\F_2$-vector space $\F_2^V$ via characteristic vectors, so that symmetric difference is vector addition, the empty set is the zero vector $0$, and $V$ itself is the all-ones vector $1$.

\begin{definition}[Ordinary subgraphs and four-vertex-connectivity]
\label{defn:ordinary}
Let $G$ be a graph. An \emph{ordinary subgraph} of $G$ is a graph $H$ with $V(H)\subseteq V(G)$ and $E(H)\subseteq\{e\in E(G):\text{both ends of $e$ lie in $V(H)$}\}$; the subgraph need not be induced. The graph $H$ is \emph{four-vertex-connected} in the sense of~\cite[Definition~6.3]{JNW21} with $k=4$ if $|V(H)|\ge 5$ (equivalently $|V(H)|>4$) and $H-U$ is connected whenever $U\subseteq V(H)$ and $|U|\le 3$.
\end{definition}

\begin{definition}[Legal systems]
\label{defn:legal}
Let $H$ be a graph. A \emph{JNW move} at a vertex $v$ of $H$ is a subset $m\subseteq V(H)$ such that $v\in m$ and $m$ contains no neighbour of $v$. Equivalently, writing $U_v(H)$ for the set of all such moves, one has
\[
U_v(H)=\bigl\{w\in\F_2^{V(H)}:w_v=1\text{ and }w_u=0\text{ for every }u\in N(v)\bigr\}.
\]
A state $T\subseteq V(H)$ is \emph{legal} if the induced subgraphs $H[T]$ and $H[V(H)\setminus T]$ are both nonempty and connected. Write $L_H$ for the set of legal states. A \emph{legal system} for $H$ is a choice of a move $m_v\in U_v(H)$ at each vertex $v$ together with a state $S$ such that every vector of the coset $S+\operatorname{span}_{\F_2}\{m_v:v\in V(H)\}$ is a legal state~\cite[Section~2]{JNW21}.
\end{definition}

\begin{definition}[Two-curvature]
\label{defn:kappa}
For a graph $H$ with $n=|V(H)|$ and $m=|E(H)|$, write $\kappa_2(H):=1-n/2+m/4$. This is the $r=2$ case of the clique-count curvature of~\cite[Definition~3.1]{JNW21}: writing $K_i$ for the set of $i$-cliques, with $|K_{-1}|=1$,
\[
\kappa_r(H)=\sum_{i=-1}^{r}(-2)^{-i-1}|K_i|,\qquad \kappa(H)=\kappa_\infty(H).
\]
In particular $\kappa_2(H)=1-|V(H)|/2+|E(H)|/4$. If $H$ is triangle-free then $K_i$ is empty for all $i\ge 3$, so $\kappa(H)=\kappa_2(H)$. This is the Charney--Davis curvature appearing in the right-angled setting; see~\cite{CD95,JNW21}.
\end{definition}

\begin{lemma}[Four-regular curvature]
\label{lem:kappa}
Let $G$ be a finite $4$-regular simple undirected graph, and write $n=|V(G)|$. Then $|E(G)|=2n$ and $\kappa_2(G)=1$.
\end{lemma}

\begin{proof}
The sum of vertex degrees equals $4n$. By the handshaking lemma this sum equals $2|E(G)|$, so $|E(G)|=2n$. Substituting into Definition~\ref{defn:kappa} gives $\kappa_2(G)=1-n/2+(2n)/4=1$.
\end{proof}

\begin{theorem}[Legal systems force nonnegative $2$-curvature]
\label{thm:curvature}
Let $G$ be a finite simple undirected graph. If $G$ admits a legal system, then $\kappa_2(G)\ge 0$. If $\kappa_2(G)=0$, then for every state $T$ in the corresponding legal coset the induced subgraph $G[T]$ is a tree.
\end{theorem}

\begin{proof}
The inequality $\kappa_2\ge 0$ in the presence of a legal system is~\cite[Theorem~6.14]{JNW21}. The argument below is the averaging proof of that theorem, specialised at the end to the hexagonal prism.

Write $M$ for the $\F_2$-span of the chosen moves $\{m_v:v\in V(G)\}$, and write $O$ for the coset $S+M$. By hypothesis every $T\in O$ is a legal state. Let $N:=|O|$. Then $N$ is a positive power of $2$, and $N\ge 2$ because each generator $m_v$ is nonzero.

Let $K=\{x_1,\ldots,x_d\}$ be a clique of $G$ on $d\ge 1$ vertices, and let $m_i$ be the chosen move at $x_i$. Because $K$ is a clique, $m_i$ contains $x_i$ and contains none of $x_1,\ldots,x_d$ except $x_i$. Restricting the characteristic vectors of $m_1,\ldots,m_d$ to the $d$ coordinates indexed by $K$ therefore yields the $d\times d$ identity matrix over $\F_2$. In particular these $d$ vectors are linearly independent, the subgroup $H:=\operatorname{span}\{m_1,\ldots,m_d\}$ of $M$ has order $2^d$, and the translation action of $H$ on $O$ is free. The orbit $O$ is thereby partitioned into $N/2^d$ classes of size $2^d$, each a coset of $H$. Within a single $H$-class, the $d$ coordinates indexed by $K$ take all $2^d$ possible values exactly once, so $K$ is contained in exactly $N/2^d$ states of $O$.

A single vertex is a clique of size $1$, so each vertex of $G$ lies in exactly $N/2$ states of $O$. An edge is a clique of size $2$, so the two endpoints of each edge lie together in exactly $N/4$ states of $O$. Summing therefore yields
\[
\sum_{T\in O}|T|=n\cdot\frac{N}{2},\qquad\sum_{T\in O}|E(G[T])|=m\cdot\frac{N}{4}.
\]

Let $T$ be a state in $O$. Then $G[T]$ is nonempty and connected. A connected graph on $k\ge 1$ vertices has at least $k-1$ edges, with equality if and only if the graph is a tree: the case $k=1$ is empty of edges, while for $k\ge 2$ a spanning tree exists by iteratively adding an edge from the current vertex-set to a new vertex, a tree on $k$ vertices has exactly $k-1$ edges, and every remaining edge is extra. Taking $k=|T|$ therefore yields $|E(G[T])|\ge |T|-1$, or equivalently $1-|T|+|E(G[T])|\ge 0$, with equality if and only if $G[T]$ is a tree.

Summing over the $N$ states of $O$ gives
\[
\begin{aligned}
0 &\le \sum_{T\in O}\bigl(1-|T|+|E(G[T])|\bigr) \\
  &= N-\sum_{T\in O}|T|+\sum_{T\in O}|E(G[T])| \\
  &= N\bigl(1-n/2+m/4\bigr)=N\kappa_2(G).
\end{aligned}
\]
Since $N>0$ one obtains $\kappa_2(G)\ge 0$. If $\kappa_2(G)=0$ then each summand vanishes, so $G[T]$ is a tree for every $T\in O$.\end{proof}

\begin{corollary}[Negative $2$-curvature obstruction]
\label{cor:negative-curvature}
If $\kappa_2(G)<0$, then $G$ admits no legal system.
\end{corollary}

\begin{proof}
This is the contrapositive of the first assertion of Theorem~\ref{thm:curvature}.
\end{proof}

For later use, let $R=C_6\times K_2$ denote the Cartesian product of the $6$-cycle and $K_2$, i.e., the hexagonal prism. It has $12$ vertices and $18$ edges, and hence
\[
\kappa_2(R)=1-\frac{12}{2}+\frac{18}{4}=-\frac12<0.
\]
Thus $R$ admits no legal system by Corollary~\ref{cor:negative-curvature}.

\begin{theorem}[Affine characterisation of legal systems]
\label{thm:affine}
Let $G$ be a finite simple undirected graph with vertex-set $V$, and for each $v\in V$ let
\[
U_v=\{w\in\F_2^V:w_v=1\text{ and }w_u=0\text{ for every }u\in N(v)\}.
\]
Then $G$ admits a legal system if and only if there exist $s\in\F_2^V$ and a linear subspace $\mathrm{Dir}\le \F_2^V$ such that
\[
s+\mathrm{Dir}\subseteq L_G
\qquad\text{and}\qquad
\mathrm{Dir}\cap U_v\neq\emptyset\quad\text{for every }v\in V.
\]
\end{theorem}

\begin{proof}
A legal system supplies moves $m_v\in U_v$ and a state $s$ whose coset under $M:=\operatorname{span}\{m_v:v\in V\}$ lies in $L_G$. Setting $\mathrm{Dir}:=M$ gives $s+\mathrm{Dir}\subseteq L_G$, and each generator $m_v$ witnesses that $\mathrm{Dir}\cap U_v$ is nonempty.

Conversely, suppose $s$ and $\mathrm{Dir}$ satisfy the two affine conditions. For each vertex $v$ choose $m_v\in\mathrm{Dir}\cap U_v$. Then $m_v$ is a move at $v$. Let $M:=\operatorname{span}\{m_v:v\in V\}$. Each generator lies in $\mathrm{Dir}$, so $M\subseteq\mathrm{Dir}$, and therefore $s+M\subseteq s+\mathrm{Dir}\subseteq L_G$. The chosen moves together with $s$ form a legal system. (If $\mathrm{Dir}=\{0\}$ then $\mathrm{Dir}\cap U_v$ is empty for every $v$, because $0\notin U_v$; any certifying subspace has dimension at least $1$. The argument permits $\mathrm{Dir}$ to be strictly larger than the span of the eventually chosen moves.)
\end{proof}

\begin{corollary}[Complement-closed affine certificate]
\label{cor:complement-certificate}
If $G$ admits a legal system, then the pair $s,\mathrm{Dir}$ in Theorem~\ref{thm:affine} may be chosen with $1\in\mathrm{Dir}$. Consequently the affine family $s+\mathrm{Dir}$ is closed under complementation.
\end{corollary}

\begin{proof}
Legality is complement-symmetric, so $L_G$ is invariant under translation by the all-ones vector $1$. Starting from a certificate $s,\mathrm{Dir}$ given by Theorem~\ref{thm:affine}, replace $\mathrm{Dir}$ by $\mathrm{Dir}+\operatorname{span}\{1\}$. Then
\[
s+(\mathrm{Dir}+\operatorname{span}\{1\})=(s+\mathrm{Dir})\cup(s+\mathrm{Dir}+1)\subseteq L_G,
\]
and the intersections with every $U_v$ remain nonempty.
\end{proof}

\section{Amalgams along induced $4$-cycles}
\label{sec:amalgam}

We first prove that a legal system on an amalgam along an induced $4$-cycle restricts to both sides. This is the obstruction-preserving operation used in the construction and is the $C=K(2,2)$ specialisation of~\cite[Remark~6.11]{JNW21}. For completeness and to make the hypotheses needed later explicit, we give a self-contained proof. The argument passes to the complement-closed affine certificate of Corollary~\ref{cor:complement-certificate} and studies its coordinate restriction to one side. If every restricted state is legal, the affine characterisation gives a legal system there directly; otherwise the induced $4$-cycle forces a small set of exceptional configurations that are eliminated by an isolation argument. Symmetry then gives the conclusion for both sides.

\begin{definition}[Induced $4$-cycle amalgam]
\label{defn:amalgam}
Let $\Gamma'$ and $\Gamma''$ be finite simple undirected graphs and let $C$ be a $4$-cycle. We call $\Gamma=\Gamma'\cup_C\Gamma''$ an \emph{induced $4$-cycle amalgam} if
\begin{enumerate}
\item $V(\Gamma')\cap V(\Gamma'')=V(C)$ and $E(\Gamma')\cap E(\Gamma'')=E(C)$;
\item the subgraph of $\Gamma$ induced by $V(C)$ is exactly $C$;
\item no edge of $\Gamma$ joins $V(\Gamma')\setminus V(C)$ to $V(\Gamma'')\setminus V(C)$.
\end{enumerate}
\end{definition}

\begin{theorem}[Induced $4$-cycle amalgam restriction]
\label{thm:amalgam}
Let $\Gamma=\Gamma'\cup_C\Gamma''$ be an induced $4$-cycle amalgam. If $\Gamma$ admits a legal system, then both $\Gamma'$ and $\Gamma''$ admit legal systems.
\end{theorem}

\begin{proof}
Label the vertices of $C$ as $a,b,c,d$ in cyclic order, and set
\[
W':=V(\Gamma')\setminus V(C),\qquad W'':=V(\Gamma'')\setminus V(C).
\]
By Definition~\ref{defn:amalgam}, the subgraphs of $\Gamma$ induced by $V(\Gamma')$ and $V(\Gamma'')$ are exactly $\Gamma'$ and $\Gamma''$, respectively, and there are no edges from $W'$ to $W''$.

By Corollary~\ref{cor:complement-certificate}, choose $s\in\F_2^{V(\Gamma)}$ and a linear subspace $\mathrm{Dir}\leq \F_2^{V(\Gamma)}$ such that
\[
\mathcal A:=s+\mathrm{Dir}\subseteq L_\Gamma,
\qquad
\mathrm{Dir}\cap U_v(\Gamma)\neq\emptyset\quad(v\in V(\Gamma)),
\qquad
1\in\mathrm{Dir}.
\]
For every vertex $v$ fix one vector $m_v\in\mathrm{Dir}\cap U_v(\Gamma)$. Thus every member of $\mathcal A$ is legal, $\mathcal A$ is closed under translation by every $m_v$, and, since $1\in\mathrm{Dir}$, $\mathcal A$ is also closed under complementation. We work throughout with this affine family; we do not require $\mathrm{Dir}$ to equal the span of the selected moves $m_v$.

\begin{claim}
The $4$-cycle $Q$ on vertices $a,b,c,d$ with edges $ab,bc,cd,da$ admits a legal system.
\end{claim}

Assign $m_a=m_c=\{a,c\}$ and $m_b=m_d=\{b,d\}$. Each of these is a move at the indicated vertex, because $\{a,c\}$ is an independent set containing $a$ (respectively $c$) and containing no neighbour of $a$ (respectively of $c$), and likewise for $\{b,d\}$. Let $S=\{a,b\}$. The $\F_2$-span of the two distinct moves has order $4$, and the coset of $S$ is $\bigl\{\{a,b\},\{b,c\},\{a,d\},\{c,d\}\bigr\}$. Each of these four sets is an edge of $Q$, hence induces a connected nonempty subgraph, and the complement of each is the opposite edge, likewise connected and nonempty.

\begin{claim}
Let $H$ be a finite simple graph, let $\mathcal B=t+D\subseteq L_H$ be an affine family of legal states, and let $n_w\in D\cap U_w(H)$. If $T\in\mathcal B$ and $w\notin T$ with no neighbour of $w$ in $T$, then $T+n_w=\{w\}$. If $w\in T$ and every neighbour of $w$ belongs to $T$, then $T+n_w=V(H)\setminus\{w\}$.
\end{claim}

In the first case $n_w$ contains $w$ and contains no neighbour of $w$, so $w$ belongs to $T+n_w$ while no neighbour of $w$ does. Thus $w$ is isolated in $H[T+n_w]$. Since $n_w\in D$, the state $T+n_w$ lies in $\mathcal B$ and is therefore nonempty and connected, forcing $T+n_w=\{w\}$. In the second case $w$ leaves the state while every neighbour of $w$ stays, so $w$ is isolated in the complement of $T+n_w$. Again $T+n_w$ is legal, hence its complement is nonempty and connected; therefore that complement equals $\{w\}$.

\begin{claim}
Let $H$ be a finite simple graph and let $\mathcal B=t+D\subseteq L_H$ be complement-closed. Suppose $D$ contains a move $n_v\in U_v(H)$ for every vertex $v$. Then $\mathcal B$ contains neither a singleton $\{x\}$ nor a co-singleton $V(H)\setminus\{x\}$ when $x$ has at least two non-neighbours.
\end{claim}

Suppose first that $\{x\}\in\mathcal B$, and let $w,z$ be distinct non-neighbours of $x$. Applying the previous claim to the state $\{x\}$ at $w$ gives $\{x\}+n_w=\{w\}$, hence $n_w=\{x,w\}$. Likewise $n_z=\{x,z\}$. Therefore $n_w+n_z=\{w,z\}\in D$, and so
\[
\{x\}+n_w+n_z=\{x,w,z\}\in\mathcal B.
\]
But $x$ is adjacent to neither $w$ nor $z$, so $x$ is isolated in the induced subgraph on this three-vertex state, contradicting legality. If instead $V(H)\setminus\{x\}\in\mathcal B$, complement-closure gives $\{x\}\in\mathcal B$, and the same contradiction applies.

\begin{claim}
The cone on $C$ from a single apex $x$ admits a legal system.
\end{claim}

Let $H$ have vertex-set $\{x,a,b,c,d\}$, with cycle edges $ab,bc,cd,da$ and apex edges $xa,xb,xc,xd$. Assign $m_x=\{x\}$, $m_a=m_c=\{a,c\}$, and $m_b=m_d=\{b,d\}$. These three distinct moves have pairwise disjoint supports, so they are linearly independent and their span has order $8$. The coset of $\{a,b,x\}$ consists of
\[
\{a,b,x\},\ \{a,b\},\ \{b,c,x\},\ \{b,c\},\ \{a,d,x\},\ \{a,d\},\ \{c,d,x\},\ \{c,d\}.
\]
For each of the four triples containing $x$, the two equator vertices are adjacent on $C$ and both are adjacent to $x$, so the induced subgraph is a triangle; its complement is the opposite equator edge. For each of the four equator edges not containing $x$, the induced subgraph is an edge; its complement is the opposite edge together with $x$, which again induces a triangle. Thus all eight states are legal.

\noindent\textbf{Step 1.}
The empty-interior side is the interface cycle, which has a legal system. If $W'$ is empty, then $V(\Gamma')=\{a,b,c,d\}$ and $\Gamma'$ equals $C$, because $C$ is induced in $\Gamma$ and $\Gamma'$ contains the four cycle edges. The first claim supplies a legal system on $\Gamma'$.

\noindent\textbf{Step 2.}
Now assume $W'$ is nonempty. Let
\[
\pi\colon\F_2^{V(\Gamma)}\longrightarrow \F_2^{V(\Gamma')}
\]
be coordinate projection. For every $x\in V(\Gamma')$, the vector $\pi(m_x)=m_x\cap V(\Gamma')$ is a move at $x$ in $\Gamma'$ and belongs to $\pi(\mathrm{Dir})$. If every member of
\[
\pi(\mathcal A)=\pi(s)+\pi(\mathrm{Dir})
\]
is a legal state of $\Gamma'$, then Theorem~\ref{thm:affine}, applied to $\pi(s)$ and $\pi(\mathrm{Dir})$, immediately gives a legal system on $\Gamma'$. Hence it remains only to treat the case in which some $S\in\mathcal A$ has
\[
S':=S\cap V(\Gamma')
\]
not legal in $\Gamma'$. Since $\mathcal A$ is complement-closed, after replacing $S$ by $V(\Gamma)\setminus S$ if necessary, we may assume that $S'$ is empty or that $\Gamma'[S']$ is disconnected.

\noindent\textbf{Step 3.}
First suppose $S\cap V(\Gamma')$ is empty. Since $S$ is legal in $\Gamma$, it is nonempty, so $S$ is a nonempty subset of $W''$. Pick $x\in W'$. Then $x\notin S$, and every neighbour of $x$ lies in $V(\Gamma')$, hence outside $S$. The isolation claim gives $S+m_x=\{x\}$, so $\{x\}\in\mathcal A$. Every vertex of $W''$ is a non-neighbour of $x$.

If $W''$ is empty, this contradicts the nonemptiness of $S$. If $|W''|\ge2$, the singleton claim is a contradiction. It remains to consider $W''=\{v\}$. If $x$ has any non-neighbour in $V(\Gamma')\setminus\{x\}$, then together with $v$ it has at least two non-neighbours, again contradicting the singleton claim. Hence $x$ is adjacent to every other vertex of $\Gamma'$, in particular to $a,b,c,d$.

Because there are no cross edges between $W'$ and $W''$, the vertex $v$ is not adjacent to $x$. Applying the isolation claim to the state $\{x\}$ at $v$ yields $\{x\}+m_v=\{v\}$, so $m_v=\{x,v\}$ and $\{v\}\in\mathcal A$. Let $k$ be the number of neighbours of $v$ in $C$. If $k\le3$, choose $z\in C$ not adjacent to $v$. Applying the isolation claim to $\{v\}$ at $z$ gives $m_z=\{v,z\}$ and $\{z\}\in\mathcal A$. Hence
\[
\{v\}+m_v+m_z=\{v,x,z\}\in\mathcal A.
\]
The only edge of $\Gamma$ on $\{v,x,z\}$ is $xz$, so $v$ is isolated, a contradiction. Thus $k=4$ and $v$ is adjacent to all four vertices of $C$.

If $W'$ contains a vertex $y\neq x$, then $y$ is adjacent to $x$ and is not adjacent to $v$. Isolation on the state $\{v\}$ at $y$ gives $m_y=\{v,y\}$. Therefore
\[
\{v\}+m_v+m_y=\{v,x,y\}\in\mathcal A,
\]
and $v$ is isolated because there are no edges from $v$ to $W'$, again a contradiction. Consequently $W'=\{x\}$. Since $x$ is adjacent to all of $C$, the graph $\Gamma'$ is exactly the cone on $C$ from $x$, which has a legal system by the cone claim. Thus the empty-restriction case always either contradicts legality of $\mathcal A$ or directly supplies a legal system on $\Gamma'$.

\noindent\textbf{Step 4.}
Now suppose $X:=S\cap V(\Gamma')$ is nonempty and $\Gamma'[X]$ is disconnected. Let $X_1,\ldots,X_k$, $k\ge2$, be its connected components. Every component $X_i$ meets $C$: otherwise $X_i\subseteq W'$, and because there are no edges from $W'$ to $W''$, the set $X_i$ would be a connected component of $\Gamma[S]$, contradicting the connectedness of the legal state $S$. Thus $C\cap S$ meets at least two distinct components of $\Gamma'[X]$, so the subgraph of the induced $4$-cycle $C$ on $C\cap S$ is disconnected. The only nonempty disconnected induced subgraphs of a $4$-cycle are its two opposite pairs. After relabelling, we may therefore assume
\[
a,c\in S,\qquad b,d\notin S,
\]
and $a,c$ lie in different components of $\Gamma'[X]$. Since $\Gamma[S]$ is connected, an $S$-path from $a$ to $c$ must pass through $W''$, so $S\cap W''\neq\emptyset$.

The moves $m_a,m_c$ contain neither $b$ nor $d$. If $m_a$ contains $c$, let $m=m_a$; if $m_c$ contains $a$, let $m=m_c$; otherwise let $m=m_a+m_c$. In every case $m\in\mathrm{Dir}$ and
\[
m\cap C=\{a,c\}.
\]
Set $T:=S+m\in\mathcal A$. Then $T\cap C=\emptyset$. Since there are no edges from $W'$ to $W''$, the induced graph $\Gamma[T]$ is the disjoint union of the subgraphs on $T\cap W'$ and $T\cap W''$. Because $T$ is a nonempty connected legal state, exactly one of these two sets is empty.

Suppose first that $T$ is a nonempty subset of $W''$. Pick $x\in W'$. No neighbour of $x$ lies in $T$, so isolation gives $T+m_x=\{x\}\in\mathcal A$. If $|W''|\ge2$, this contradicts the singleton claim. The case $W''=\emptyset$ is impossible because $S\cap W''\neq\emptyset$. If $|W''|=1$, the residual configuration is exactly the singleton configuration analysed in Step~3: either the singleton claim gives a contradiction or $\Gamma'$ is the cone on $C$, which has a legal system.

Suppose instead that $T$ is a nonempty subset of $W'$. Choose $y\in W''$, which exists because $S\cap W''\neq\emptyset$. No neighbour of $y$ lies in $T$, so isolation gives $T+m_y=\{y\}\in\mathcal A$. The non-neighbour set of $y$ contains all of $W'$. Hence $|W'|\ge2$ contradicts the singleton claim. If $W'=\{x\}$, then $T=\{x\}\in\mathcal A$. Now the non-neighbour set of $x$ contains $W''$. If $|W''|\ge2$, the singleton claim again gives a contradiction; if $|W''|=1$, we are again in the residual cone configuration of Step~3. Thus this case also either contradicts legality or yields a legal system on $\Gamma'$.

We have proved that $\Gamma'$ admits a legal system. By symmetry the same argument gives a legal system on $\Gamma''$.
\end{proof}

\begin{corollary}[Obstruction inheritance]
\label{cor:amalgam-obstruction}
Let $\Gamma=\Gamma'\cup_C\Gamma''$ be an induced $4$-cycle amalgam. If either $\Gamma'$ or $\Gamma''$ admits no legal system, then $\Gamma$ admits no legal system.
\end{corollary}

\begin{proof}
This is the contrapositive of Theorem~\ref{thm:amalgam}.
\end{proof}

By induction, the same obstruction persists through any finite sequence of induced $4$-cycle amalgams.

\begin{remark}
\label{rem:amalgam-scope}
Theorem~\ref{thm:amalgam}, with the convention of Definition~\ref{defn:amalgam}, is the specialisation of~\cite[Remark~6.11]{JNW21} to a cocktail-party interface $C=K(2,2)$. The source's~\cite[Lemma~6.10]{JNW21} is the special case in which one side is the cone on a $4$-cycle. Requiring that the induced subgraph on $\{a,b,c,d\}$ equal $C$ records that the interface is the cocktail-party graph itself. Extra edges of a side among its interior, or from its interior to $C$, remain. Because $C$ is induced in $\Gamma$ and there are no cross edges between interiors, each side is an induced subgraph of $\Gamma$.
\end{remark}

\section{The prism--cap host}
\label{sec:host}

This section constructs the $33$-vertex prism--cap host and verifies the structural properties used in the final argument. Figure~\ref{fig:prism-cap} displays the two building blocks, while Figure~\ref{fig:full-host} shows the complete $33$-vertex graph obtained after all three caps are attached. We then check its order, regularity, girth, and the exact induced-$4$-cycle interfaces needed for Theorem~\ref{thm:amalgam}. The connectivity proof is separated from these local checks: deletions in the prism and in a cap interior are controlled first, and the only delicate case---three deleted prism vertices---is reduced to three surviving interface blocks joined by the remaining prism links.

\begin{definition}[The prism--cap graph]
\label{defn:G}
Let $R$ be the hexagonal prism with vertex-set $\{a_0,a_1,a_2,a_3,a_4,a_5,b_0,b_1,b_2,b_3,b_4,b_5\}$ and with edge-set consisting of the twelve rail edges $a_ia_{i+1}$ and $b_ib_{i+1}$ together with the six spokes $a_ib_i$, where indices of rails are read modulo $6$. Under this labelling, $R=C_6\times K_2$. Let $Q_0,Q_1,Q_2$ be the three $4$-cycles
\[
Q_0=(a_0,a_1,b_1,b_0),\qquad Q_1=(a_2,a_3,b_3,b_2),\qquad Q_2=(a_4,a_5,b_5,b_4),
\]
written in cyclic order. For each $t\in\{0,1,2\}$ let $T_t$ be the graph obtained from the $4$-cycle $Q_t$ by adding seven fresh vertices
\[
X_t=\{x_{t,0},x_{t,1},x_{t,2}\},\qquad
Y_t=\{y_{t,0},y_{t,1},y_{t,2},y_{t,3}\}.
\]
The graph $T_t$ contains every edge $x_{t,r}y_{t,j}$, the four matching edges joining $Q_t[j]$ to $y_{t,j}$, and no other new edges. For $T_0$, write the interior vertices as
\[
x_{00},x_{01},x_{02},\qquad y_{00},y_{01},y_{02},y_{03},
\]
and use the analogous compact names with first index $1$ and $2$ for $T_1$ and $T_2$. Let $G$ be the simple graph obtained from the disjoint union of $R$, $T_0$, $T_1$, $T_2$ by identifying, for each $t$, the boundary $4$-cycle of $T_t$ with the cycle $Q_t$ of $R$ along the displayed cyclic order. Write $I_t:=X_t\cup Y_t$ for the open interior of $T_t$. Explicitly the matching edges are
\begin{align*}
a_0y_{00},&\ a_1y_{01},\ b_1y_{02},\ b_0y_{03},\\
a_2y_{10},&\ a_3y_{11},\ b_3y_{12},\ b_2y_{13},\\
a_4y_{20},&\ a_5y_{21},\ b_5y_{22},\ b_4y_{23}.
\end{align*}
\end{definition}

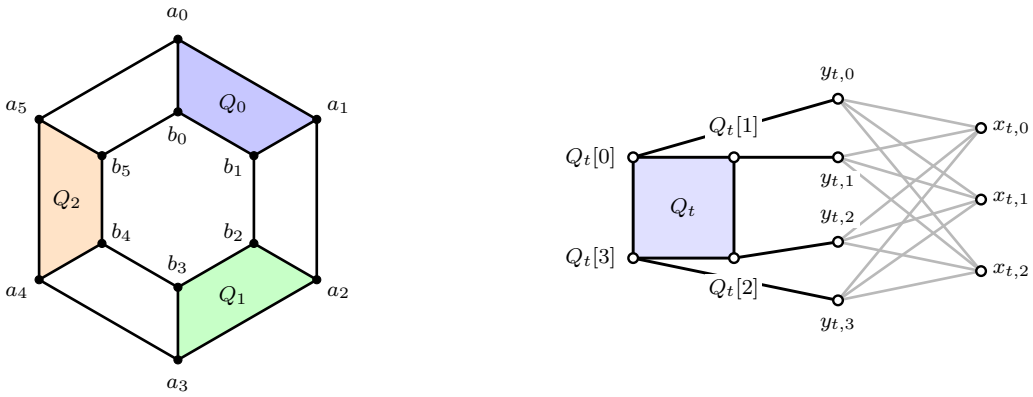
\begin{figure}[!htbp]
\centering
\begin{minipage}[c]{0.45\textwidth}
\centering
\begin{tikzpicture}[scale=0.80]
  \def\Ra{2.65}
  \def\Rb{1.45}
  \fill[blue!22] (90:\Ra) -- (30:\Ra) -- (30:\Rb) -- (90:\Rb) -- cycle;
  \fill[green!22] (-30:\Ra) -- (-90:\Ra) -- (-90:\Rb) -- (-30:\Rb) -- cycle;
  \fill[orange!22] (-150:\Ra) -- (150:\Ra) -- (150:\Rb) -- (-150:\Rb) -- cycle;
  \foreach \ang in {30,90,150,-150,-90,-30} {
    \draw[line width=1pt] (\ang:\Ra) -- (\ang:\Rb);
  }
  \draw[line width=1pt] (90:\Ra) -- (30:\Ra) -- (-30:\Ra) -- (-90:\Ra) -- (-150:\Ra) -- (150:\Ra) -- cycle;
  \draw[line width=1pt] (90:\Rb) -- (30:\Rb) -- (-30:\Rb) -- (-90:\Rb) -- (-150:\Rb) -- (150:\Rb) -- cycle;
  \foreach \ang/\lab in {90/a_0,30/a_1,-30/a_2,-90/a_3,-150/a_4,150/a_5} {
    \filldraw (\ang:\Ra) circle (1.9pt);
    \node[font=\scriptsize,fill=white,inner sep=1pt] at (\ang:{\Ra+0.42}) {$\lab$};
  }
  \foreach \ang/\lab in {90/b_0,30/b_1,-30/b_2,-90/b_3,-150/b_4,150/b_5} {
    \filldraw (\ang:\Rb) circle (1.9pt);
    \node[font=\scriptsize,fill=white,inner sep=1pt] at (\ang:{\Rb-0.38}) {$\lab$};
  }
  \node[font=\scriptsize,inner sep=0pt] at (60:1.83) {$Q_0$};
  \node[font=\scriptsize,inner sep=0pt] at (-60:1.83) {$Q_1$};
  \node[font=\scriptsize,inner sep=0pt] at (180:1.83) {$Q_2$};
\end{tikzpicture}
\end{minipage}\hfill
\begin{minipage}[c]{0.52\textwidth}
\centering
\begin{tikzpicture}[scale=0.86]
  % Boundary square Q_t on the left.
  \coordinate (c0) at (0,1.55);
  \coordinate (c1) at (1.55,1.55);
  \coordinate (c2) at (1.55,0);
  \coordinate (c3) at (0,0);
  % Four Y-vertices form a separated middle column.
  \coordinate (y0) at (3.15,2.45);
  \coordinate (y1) at (3.15,1.55);
  \coordinate (y2) at (3.15,0.25);
  \coordinate (y3) at (3.15,-0.65);
  % Three X-vertices form the right column of K_{3,4}.
  \coordinate (x0) at (5.35,2.00);
  \coordinate (x1) at (5.35,0.90);
  \coordinate (x2) at (5.35,-0.20);

  % Draw K_{3,4} first so node labels remain unobstructed.
  \foreach \yy in {y0,y1,y2,y3} {
    \foreach \xx in {x0,x1,x2} {
      \draw[gray!55,line width=1pt] (\yy) -- (\xx);
    }
  }
  \fill[blue!12] (c0) -- (c1) -- (c2) -- (c3) -- cycle;
  \draw[line width=1.1pt] (c0) -- (c1) -- (c2) -- (c3) -- cycle;
  \draw[line width=1.1pt] (c0) -- (y0);
  \draw[line width=1.1pt] (c1) -- (y1);
  \draw[line width=1.1pt] (c2) -- (y2);
  \draw[line width=1.1pt] (c3) -- (y3);

  \foreach \p in {c0,c1,c2,c3,y0,y1,y2,y3,x0,x1,x2} {
    \filldraw[fill=white,line width=0.9pt] (\p) circle (2.1pt);
  }

  \node[font=\scriptsize,fill=white,inner sep=1.2pt,left=5pt] at (c0) {$Q_t[0]$};
  \node[font=\scriptsize,fill=white,inner sep=1.2pt,above=6pt] at (c1) {$Q_t[1]$};
  \node[font=\scriptsize,fill=white,inner sep=1.2pt,below=6pt] at (c2) {$Q_t[2]$};
  \node[font=\scriptsize,fill=white,inner sep=1.2pt,left=5pt] at (c3) {$Q_t[3]$};
  \node[font=\scriptsize,inner sep=0pt] at (0.775,0.775) {$Q_t$};

  \node[font=\scriptsize,fill=white,inner sep=1.2pt,above=5pt] at (y0) {$y_{t,0}$};
  \node[font=\scriptsize,fill=white,inner sep=1.2pt,below=5pt] at (y1) {$y_{t,1}$};
  \node[font=\scriptsize,fill=white,inner sep=1.2pt,above=5pt] at (y2) {$y_{t,2}$};
  \node[font=\scriptsize,fill=white,inner sep=1.2pt,below=5pt] at (y3) {$y_{t,3}$};

  \node[font=\scriptsize,fill=white,inner sep=1.2pt,right=3pt] at (x0) {$x_{t,0}$};
  \node[font=\scriptsize,fill=white,inner sep=1.2pt,right=3pt] at (x1) {$x_{t,1}$};
  \node[font=\scriptsize,fill=white,inner sep=1.2pt,right=3pt] at (x2) {$x_{t,2}$};
\end{tikzpicture}
\end{minipage}
\caption{The hexagonal prism $R$ with its three disjoint interface $4$-cycles $Q_0,Q_1,Q_2$ (left), and one representative cap $T_t$ whose $K_{3,4}$ interior is matched to the boundary cycle $Q_t$ (right).}

\label{fig:prism-cap}
\end{figure}

\begin{proposition}[Structure and connectivity of the host]
\label{prop:host}
Let $G$ be the graph of Definition~\ref{defn:G}. Then the following hold.
\begin{enumerate}
\item $G$ has $33$ vertices and $66$ edges, is $4$-regular and triangle-free, has girth four, and satisfies $\kappa_2(G)=1$.
\item $G$ is four-vertex-connected in the sense of~\cite[Definition~6.3]{JNW21} with $k=4$: $|V(G)|=33\ge5$, and $G-U$ is connected for every $U\subseteq V(G)$ with $|U|\le3$.
\item For each $t\in\{0,1,2\}$, the pair $(R,T_t)$ forms an exact induced-$4$-cycle interface: $V(R)\cap V(T_t)=V(Q_t)$ and $E(R)\cap E(T_t)=E(Q_t)$; moreover $Q_t$ is induced in $G$, and no edge joins $V(T_t)\setminus V(Q_t)$ to $V(G)\setminus V(T_t)$.
\end{enumerate}
\end{proposition}

\begin{proof}
The three sets $V(Q_0)=\{a_0,a_1,b_0,b_1\}$, $V(Q_1)=\{a_2,a_3,b_2,b_3\}$, $V(Q_2)=\{a_4,a_5,b_4,b_5\}$ are pairwise disjoint and their union is $V(R)$. The seven interior vertices of each $T_t$ are chosen fresh and pairwise disjoint from $V(R)$ and from the interiors of the other caps. Thus $|V(G)|=12+3\cdot 7=33$. The thirty-three names $a_0,\ldots,a_5$, $b_0,\ldots,b_5$, $x_{00},x_{01},x_{02}$, $y_{00},y_{01},y_{02},y_{03}$, $x_{10},x_{11},x_{12}$, $y_{10},y_{11},y_{12},y_{13}$, $x_{20},x_{21},x_{22}$, $y_{20},y_{21},y_{22},y_{23}$ are distinct.

\noindent\textbf{Step 1.}
The edge-set is simple and $4$-regular, and each interface is an exact induced $4$-cycle. The edges of $G$ are exactly: the eighteen edges of $R$; for each $t$, the twelve edges of the complete bipartite graph between $X_t$ and $Y_t$; and for each $t$, the four matching edges $Q_t[j]y_{t,j}$. The four boundary edges of each $T_t$ are identified with the four edges of $Q_t$ already present in $R$, and are not duplicated. There are no loops in this list, and no two listed pairs coincide after the boundary identifications. Hence $G$ is simple, and $|E(G)|=18+3\cdot(12+4)=66$.

The open neighbourhoods in $G$ are as follows. For each $t$ and each $r\in\{0,1,2\}$, the vertex $x_{t,r}$ is adjacent exactly to the four vertices of $Y_t$, so has degree $4$. For each $t$ and each $j\in\{0,1,2,3\}$, the vertex $y_{t,j}$ is adjacent exactly to the three vertices of $X_t$ and to $Q_t[j]$, so has degree $4$. In $R$ every vertex has degree $3$ (two rail edges and one spoke). Each vertex of $R$ lies in exactly one of $Q_0,Q_1,Q_2$, and therefore receives exactly one matching edge from the corresponding cap. Thus every vertex of $R$ has degree $4$ in $G$. Explicitly,
\begin{align*}
N(a_0)&=\{a_1,a_5,b_0,y_{00}\},&
N(a_1)&=\{a_0,a_2,b_1,y_{01}\},\\
N(a_2)&=\{a_1,a_3,b_2,y_{10}\},&
N(a_3)&=\{a_2,a_4,b_3,y_{11}\},\\
N(a_4)&=\{a_3,a_5,b_4,y_{20}\},&
N(a_5)&=\{a_0,a_4,b_5,y_{21}\},\\
N(b_0)&=\{a_0,b_1,b_5,y_{03}\},&
N(b_1)&=\{a_1,b_0,b_2,y_{02}\},\\
N(b_2)&=\{a_2,b_1,b_3,y_{13}\},&
N(b_3)&=\{a_3,b_2,b_4,y_{12}\},\\
N(b_4)&=\{a_4,b_3,b_5,y_{23}\},&
N(b_5)&=\{a_5,b_0,b_4,y_{22}\}.
\end{align*}
Hence $G$ is $4$-regular. Lemma~\ref{lem:kappa} gives $\kappa_2(G)=1$.

Fix $t\in\{0,1,2\}$. By construction $V(T_t)=V(Q_t)\cup X_t\cup Y_t$, and $V(R)\cap V(T_t)=V(Q_t)$. The edges of $T_t$ are the four edges of $Q_t$, the twelve edges of the complete bipartite graph on $X_t\cup Y_t$, and the four matching edges. Within the prism, the only possible diagonals of $Q_t$ are $a_{2t}b_{2t+1}$ and $a_{2t+1}b_{2t}$; neither is a rail edge or a spoke. Hence $R[V(Q_t)]=Q_t$ and $E(R)\cap E(T_t)=E(Q_t)$.

The $4$-cycle $Q_t$ is induced in $G$: $G$ has no edge among $V(Q_t)$ other than those of $R$ or of $T_t$, and both of those graphs induce exactly the $4$-cycle on $V(Q_t)$. In particular neither diagonal of $Q_t$ is present. There is no edge of $G$ with one end in the open cap $I_t:=X_t\cup Y_t$ and the other end in $V(G)\setminus V(T_t)$. The only edges incident to a vertex of $X_t$ are the four edges to $Y_t$. The only edge incident to $y_{t,j}$ other than the three edges to $X_t$ is the matching edge to $Q_t[j]$, and $Q_t[j]$ lies in $V(Q_t)\subseteq V(T_t)$. Interiors of distinct caps are disjoint and no edge between them is added.

\noindent\textbf{Step 2.}
The graph $G$ is triangle-free and has girth four. First $R$ is triangle-free. A triangle contained in one rail would be a triangle in a $6$-cycle, which does not exist. A triangle using vertices of both rails must use at least one spoke. Two adjacent rail vertices $a_i,a_{i+1}$ have spokes to the distinct vertices $b_i,b_{i+1}$, and $a_i$ is not adjacent to $b_{i+1}$, while $a_{i+1}$ is not adjacent to $b_i$. A triple $\{a_i,b_i,a_j\}$ with $j\notin\{i\}$ has the single spoke $a_ib_i$ and at most the rail edge $a_ia_j$; it does not contain a third edge closing a triangle. The same holds with the roles of the rails exchanged. Thus $R$ has no triangle.

Next $T_t$ is triangle-free. The complete bipartite graph on $X_t\cup Y_t$ has no odd cycle, hence no triangle. The boundary $4$-cycle has no triangle. A triangle containing a matching edge $Q_t[j]y_{t,j}$ would require a common neighbour of $Q_t[j]$ and $y_{t,j}$ in $T_t$. In $T_t$ the neighbours of $Q_t[j]$ other than $y_{t,j}$ are the two boundary neighbours of $Q_t[j]$ on the $4$-cycle, while the neighbours of $y_{t,j}$ other than $Q_t[j]$ are the three vertices of $X_t$. These two sets are disjoint, so no such common neighbour exists.

In the amalgam $G$ there is no edge joining $I_t$ to $V(G)\setminus V(T_t)$, by Step~1. Therefore a triangle of $G$ cannot use a vertex of $I_t$ together with a vertex outside $T_t$. Any triangle that uses a vertex of $I_t$ is then a triangle of $T_t$, which does not exist. Any triangle contained in $V(R)$ is a triangle of $R$, which does not exist. Thus $G$ is triangle-free. The $4$-cycle $Q_0$ is a cycle of $G$, so the girth of $G$ is at most $4$. Combined with triangle-freeness, the girth is exactly $4$.

\noindent\textbf{Step 3.}
Deleting at most two vertices of the prism leaves it connected. We claim that $R-U$ is connected for every $U\subseteq V(R)$ with $|U|\le 2$. If $U$ is empty, the $a$-rail is a $6$-cycle and each spoke $a_ib_i$ joins it to the $b$-rail, so $R$ is connected. If $U=\{a_i\}$, the $b$-rail remains a $6$-cycle, and every surviving $a_k$ retains the spoke $a_kb_k$, so $R-U$ is connected. The case $U=\{b_i\}$ is the same with the rails exchanged.

Now let $|U|=2$. If both vertices of $U$ lie on the $a$-rail, the $b$-rail remains a $6$-cycle and every surviving $a$-vertex retains its spoke. If both vertices of $U$ lie on the $b$-rail, the same holds with the rails exchanged. If $U=\{a_i,b_i\}$, each rail minus one vertex is a path on five vertices, and the five remaining spokes join corresponding vertices of these two paths, so $R-U$ is connected. If $U=\{a_i,b_j\}$ with $i\neq j$, each rail minus one vertex is again a path on five vertices, and the remaining spokes are the four spokes at indices other than $i$ and $j$. The unique vertex of the $a$-path with no remaining spoke is $a_j$. The $a$-rail minus $a_i$ is a path of length $4$, so $a_j$ has degree $1$ or $2$ in that path and in particular is not isolated in that path; thus $a_j$ is adjacent in $R-U$ to at least one $a$-vertex which still has its spoke. Dually, $b_i$ is not isolated on the $b$-path. Therefore the two paths together with the four remaining spokes form a connected graph, and $R-U$ is connected.

\noindent\textbf{Step 4.}
No component of $G-U$ can be trapped in a cap interior. Let $U\subseteq V(G)$ with $|U|\le 3$, and fix $t\in\{0,1,2\}$. Write $c_j$ for the vertex $Q_t[j]$. The four matching edges $c_jy_{t,j}$ are pairwise vertex-disjoint (they involve eight distinct vertices). Since $|U|\le 3$, at least one of these four pairs is disjoint from $U$. That matching edge therefore survives in $G-U$, and in particular at least one vertex of $Q_t$ survives.

If $X_t\subseteq U$, then $|U|\ge 3$, hence $U=X_t$. In this subcase every vertex of $Y_t$ and every vertex of $Q_t$ survives. Each $y_{t,j}$ is adjacent to $c_j$, and $Q_t$ is a surviving $4$-cycle, so every surviving vertex of $I_t$ (namely every vertex of $Y_t$) is adjacent to the surviving boundary.

If $X_t$ is not contained in $U$, then some vertex $x_{t,r}$ of $X_t$ survives. The set $Y_t$ has four vertices and $|U|\le 3$, so some vertex of $Y_t$ survives. The surviving vertices of $I_t$ therefore include at least one vertex of $X_t$ and at least one vertex of $Y_t$. Every remaining $x$-vertex is adjacent to every remaining $y$-vertex, so the surviving subset of $I_t$ induces a connected subgraph of $T_t$. Among the four pairwise vertex-disjoint matching edges of $T_t$, at least one pair is disjoint from $U$; that surviving matching edge has its $y$-end in this connected subgraph of $I_t$ and its boundary-end on $Q_t$. Thus every surviving vertex of $I_t$ lies in a component of $G-U$ that already contains a surviving vertex of $Q_t$, and in particular no component of $G-U$ is contained in $I_t$.

\noindent\textbf{Step 5.}
The graph $G-U$ is connected for $|U|\le 3$. Write $U_R:=U\cap V(R)$. By Step~4, every surviving interior vertex of every cap lies in a component of $G-U$ that already contains a surviving vertex of $V(R)$, whenever any interior vertex of that cap survives. It remains to show that all surviving vertices of $V(R)$ lie in a single component of $G-U$. The induced subgraph of $G$ on $V(R)$ equals $R$: every edge of $G$ with both ends in $V(R)$ is an edge of $R$, because the only new edges of the caps are either boundary edges already in $R$ or edges incident to an interior vertex.

If $|U_R|\le 2$, then $R-U_R$ is connected by Step~3. Therefore the surviving vertices of $V(R)$ induce a connected subgraph of $G-U$. Combined with Step~4, the graph $G-U$ is connected.

It remains to treat $|U_R|=3$. Then $U$ is contained in $V(R)$, so no interior vertex of any cap is deleted. The three squares $Q_0,Q_1,Q_2$ partition $V(R)$. Each has four vertices, and $|U_R|=3$, so none of the three squares is entirely deleted. For any two surviving vertices $c,c'$ of a fixed $Q_t$, the path $c$--$y(c)$--$x_{t,0}$--$y(c')$--$c'$ lies in $G-U$, where $y(c)$ denotes the unique matching neighbour of $c$ in $Y_t$: both matching neighbours and $x_{t,0}$ survive because $U$ is contained in $V(R)$. Thus, for each $t$, all surviving vertices of $Q_t$ lie in a single connected subgraph of $G-U$, which we call the block-component of $Q_t$.

The prism supplies three block-links, each consisting of two edges: $Q_0$ to $Q_1$ by the edges $a_1a_2$ and $b_1b_2$; $Q_1$ to $Q_2$ by the edges $a_3a_4$ and $b_3b_4$; $Q_2$ to $Q_0$ by the edges $a_5a_0$ and $b_5b_0$. The two edges of any one link are vertex-disjoint. The three links have pairwise disjoint endpoint-sets, and these three $4$-sets partition $V(R)$. To destroy one link (that is, to remove at least one endpoint from each of its two edges) therefore requires at least two deletions. To destroy two distinct links requires at least four deletions, and those deletions lie in disjoint vertex-sets. Since $|U_R|=3$, at most one of the three links is destroyed. The auxiliary graph whose three vertices are the three blocks and whose three edges are the three links is a $3$-cycle; after deletion of at most one edge it remains connected. Each block-component is nonempty. Therefore all surviving vertices of $V(R)$ lie in a single component of $G-U$. Every interior vertex survives: each $y_{t,j}$ is adjacent to $Q_t[j]$, and each $x_{t,r}$ is adjacent to every $y_{t,j}$, so every interior vertex of $T_t$ is attached to the block-component of $Q_t$. Hence $G-U$ is connected.

Since $|V(G)|=33\ge 5$, the graph $G$ is four-vertex-connected.
\end{proof}

\section{Proof of Theorem~\ref{thm:main}}
\label{sec:proof}

We now prove Theorem~\ref{thm:main}. The remaining pieces are the hereditary collapse of $4$-connected ordinary subgraphs in a $4$-regular host and the three sequential amalgams that carry the prism obstruction onto $G$.

\begin{lemma}[Four-regular hereditary collapse]
\label{lem:collapse}
Let $G$ be a finite simple undirected graph that is connected and $4$-regular. Let $H$ be an ordinary subgraph of $G$ that is four-vertex-connected in the sense of Definition~\ref{defn:ordinary}. Then $H$ equals $G$: $V(H)=V(G)$ and $E(H)=E(G)$.
\end{lemma}

\begin{proof}
Let $v$ be an arbitrary vertex of $H$, let $d$ be the degree of $v$ in $H$, and let $N$ be the set of neighbours of $v$ in $H$, so $|N|=d$. Suppose $d\le 3$. Set $U:=N$, so $|U|\le 3$. The vertex $v$ belongs to $H-U$. Every neighbour of $v$ in $H$ lies in $U$, so $v$ is isolated in $H-U$. If there exists a vertex $w$ in $V(H)$ that does not lie in $U\cup\{v\}$, then $H-U$ contains at least the two vertices $v$ and $w$, and $v$ is a connected component of $H-U$ distinct from the component containing $w$, contradicting that $H-U$ is connected. Therefore $V(H)=U\cup\{v\}$, whence $|V(H)|=d+1\le 4$. This contradicts $|V(H)|\ge 5$. Hence $d\ge 4$. Since $v$ was arbitrary, every vertex of $H$ has degree at least $4$ in $H$.

Let $v$ be an arbitrary vertex of $H$. Because $G$ is $4$-regular, the degree of $v$ in $G$ is exactly $4$. Every edge of $H$ incident to $v$ is an edge of $G$ incident to $v$, so the degree of $v$ in $H$ is at most $4$. Combined with the previous paragraph, the degree of $v$ in $H$ is exactly $4$. Therefore the four $G$-edges incident to $v$ all belong to $H$, and every $G$-neighbour of $v$ belongs to $V(H)$.

Consequently there is no edge of $G$ with one end in $V(H)$ and the other end in $V(G)\setminus V(H)$. Hence $V(H)$ is a union of connected components of $G$. The set $V(H)$ is nonempty because $|V(H)|\ge 5$. The graph $G$ is connected by hypothesis. A nonempty union of connected components of a connected graph must be the full vertex-set, so $V(H)=V(G)$. Every edge of $G$ is incident to a vertex of $H$ and therefore belongs to $H$, so $E(H)=E(G)$. Thus $H=G$.
\end{proof}

\begin{remark}
\label{rem:collapse-scope}
Lemma~\ref{lem:collapse} uses the convention $|V|\ge 5$ from Definition~\ref{defn:ordinary}, equivalently the numbering $|W|>4$ of~\cite[Definition~6.3]{JNW21} for integer cardinalities. The literal deletion condition ``$H-U$ connected for every $|U|\le 3$'' without $|V|\ge 5$ would count $K_4$ as four-vertex-connected, and $K_4$ is a proper ordinary subgraph of $K_5$; that degeneracy does not arise for $4$-regular graphs on fewer than $5$ vertices (a $4$-regular graph has at least $5$ vertices by the handshaking lemma) and does not arise for triangle-free subgraphs on $4$ vertices (the only $4$-vertex graph meeting the literal deletion condition is $K_4$, which contains triangles).
\end{remark}

\begin{proof}[Proof of Theorem~\ref{thm:main}]
Let $G$ be the graph of Definition~\ref{defn:G}. Proposition~\ref{prop:host} supplies that $G$ is a simple graph on $33$ vertices with $66$ edges, that $G$ is $4$-regular, triangle-free, of girth four, and four-vertex-connected, that $\kappa_2(G)=1$, and that each pair $(R,T_t)$ is an exact $4$-cycle amalgam interface.

\noindent\textbf{Step 1.}
The seed prism admits no legal system by Corollary~\ref{cor:negative-curvature} and the computation in Section~\ref{sec:prelim}.

\noindent\textbf{Step 2.}
The host is the third of three exact induced $4$-cycle amalgams. Write $G_0:=R$. Let $G_1$ be the amalgam $G_0\cup_{Q_0}T_0$ in the sense of Definition~\ref{defn:amalgam}: the vertex-set is $V(G_0)\cup V(T_0)$, the edge-set is $E(G_0)\cup E(T_0)$, the two vertex-sets meet exactly at $V(Q_0)$, and the two edge-sets meet exactly at $E(Q_0)$. Let $G_2$ be the amalgam $G_1\cup_{Q_1}T_1$ in the same sense. Let $G_3$ be the amalgam $G_2\cup_{Q_2}T_2$ in the same sense.

The vertex-set of $G_3$ is $V(R)\cup I_0\cup I_1\cup I_2$, which is $V(G)$. The edge-set of $G_3$ is $E(R)\cup E(T_0)\cup E(T_1)\cup E(T_2)$. The four boundary edges of each $T_t$ already lie in $E(R)$, so this union is exactly the sixty-six-edge set of Proposition~\ref{prop:host}. Thus $G_3=G$.

It remains to check that each of the three unions satisfies Definition~\ref{defn:amalgam}. First consider $G_1=G_0\cup_{Q_0}T_0$, with $\Gamma'=R$, $\Gamma''=T_0$, and $C=Q_0$. By the exact-interface clause of Proposition~\ref{prop:host} at $t=0$, one has $V(R)\cap V(T_0)=V(Q_0)$ and $E(R)\cap E(T_0)=E(Q_0)$. The same clause records that $Q_0$ is induced in $G$. The graph $G_1$ is the subgraph of $G$ on $V(R)\cup I_0$: every edge of $G$ with both ends in $V(R)\cup I_0$ is an edge of $R$ or of $T_0$, because the only edges of $G$ incident to $I_1$ (respectively $I_2$) have their other end in $V(T_1)$ (respectively $V(T_2)$), by the no-cross-edge clause of Proposition~\ref{prop:host}, and those other ends do not both lie in $V(R)\cup I_0$ unless the edge already belongs to $R$ or $T_0$. In particular the edges of $G_1$ among $V(Q_0)$ are among the edges of $G$ among $V(Q_0)$, which are exactly the four cycle edges of $Q_0$. Thus $Q_0$ is induced in $G_1$. There is no edge of $G_1$ from $I_0=V(T_0)\setminus V(Q_0)$ to $V(R)\setminus V(Q_0)$, because any such edge would be an edge of $G$ from $I_0$ to $V(G)\setminus V(T_0)$, contradicting the no-cross-edge clause at $t=0$. Therefore $G_1$ is an amalgam along the induced $4$-cycle $Q_0$ in the sense of Definition~\ref{defn:amalgam}.

Next consider $G_2=G_1\cup_{Q_1}T_1$, with $\Gamma'=G_1$, $\Gamma''=T_1$, and $C=Q_1$. The vertex-set of $G_1$ is $V(R)\cup I_0$, and the vertex-set of $T_1$ is $V(Q_1)\cup I_1$. These interiors $I_0$ and $I_1$ are disjoint from each other and from $V(R)$ by Proposition~\ref{prop:host}, and $V(Q_1)\subseteq V(R)$, so $\bigl(V(R)\cup I_0\bigr)\cap\bigl(V(Q_1)\cup I_1\bigr)=V(Q_1)$. The edges of $T_1$ other than the four cycle edges of $Q_1$ are incident to $I_1$, hence do not lie in $E(G_1)$. The edges of $G_1$ among $V(Q_1)$ are the edges of $R$ among $V(Q_1)$, because no edge of $T_0$ is incident to $V(Q_1)$: $V(T_0)=V(Q_0)\cup I_0$ is disjoint from $V(Q_1)$. The edges of $R$ among $V(Q_1)$ are exactly the four cycle edges of $Q_1$, by the exact-interface clause of Proposition~\ref{prop:host} at $t=1$. Thus $E(G_1)\cap E(T_1)=E(Q_1)$. The $4$-cycle $Q_1$ is induced in $G$ by that same clause, and the edges of $G_2$ among $V(Q_1)$ are among the edges of $G$ among $V(Q_1)$, so $Q_1$ is induced in $G_2$. There is no edge of $G_2$ from $I_1=V(T_1)\setminus V(Q_1)$ to $V(G_1)\setminus V(Q_1)$: any such edge would be an edge of $G$ from $I_1$ to $V(G)\setminus V(T_1)$, since $V(G_1)\setminus V(Q_1)=\bigl(V(R)\setminus V(Q_1)\bigr)\cup I_0$ is contained in $V(G)\setminus V(T_1)$, contradicting the no-cross-edge clause at $t=1$. Therefore $G_2$ is an amalgam along the induced $4$-cycle $Q_1$.

Finally consider $G_3=G_2\cup_{Q_2}T_2$, with $\Gamma'=G_2$, $\Gamma''=T_2$, and $C=Q_2$. The vertex-set of $G_2$ is $V(R)\cup I_0\cup I_1$, and the vertex-set of $T_2$ is $V(Q_2)\cup I_2$. The three interiors are pairwise disjoint and disjoint from $V(R)$, and $V(Q_2)\subseteq V(R)$, so $\bigl(V(R)\cup I_0\cup I_1\bigr)\cap\bigl(V(Q_2)\cup I_2\bigr)=V(Q_2)$. The edges of $T_2$ other than the four cycle edges of $Q_2$ are incident to $I_2$, hence do not lie in $E(G_2)$. The edges of $G_2$ among $V(Q_2)$ are the edges of $R$ among $V(Q_2)$, because no edge of $T_0$ or of $T_1$ is incident to $V(Q_2)$. The edges of $R$ among $V(Q_2)$ are exactly the four cycle edges of $Q_2$, by the exact-interface clause at $t=2$. Thus $E(G_2)\cap E(T_2)=E(Q_2)$. The $4$-cycle $Q_2$ is induced in $G$, hence induced in $G_3$. There is no edge of $G_3$ from $I_2$ to $V(G_2)\setminus V(Q_2)$, because $V(G_2)\setminus V(Q_2)$ is contained in $V(G)\setminus V(T_2)$ and the no-cross-edge clause at $t=2$ forbids any edge from $I_2$ to $V(G)\setminus V(T_2)$. Therefore $G_3$ is an amalgam along the induced $4$-cycle $Q_2$.

\noindent\textbf{Step 3.}
The host admits no legal system. Apply Corollary~\ref{cor:amalgam-obstruction} successively to the three amalgams. By Step~1, $G_0=R$ admits no legal system. The first amalgam of Step~2 therefore yields that $G_1$ admits no legal system. The second amalgam yields that $G_2$ admits no legal system. The third amalgam yields that $G_3$ admits no legal system. By Step~2 one has $G_3=G$, so $G$ admits no legal system.

\noindent\textbf{Step 4.}
It remains to verify the hereditary assertion. Proposition~\ref{prop:host} gives that $G$ is $4$-regular, triangle-free, of girth four, and four-vertex-connected, with $\kappa_2(G)=1$. By triangle-freeness and Definition~\ref{defn:kappa}, $\kappa(G)=1$. By Step~3, $G$ admits no legal system. Finally, Lemma~\ref{lem:collapse} shows that every four-vertex-connected ordinary subgraph of $G$ equals $G$; hence no such subgraph admits a legal system. This proves all assertions of Theorem~\ref{thm:main}.
\end{proof}

\section{Conclusion}
\label{sec:conclusion}

We construct a $33$-vertex, $4$-regular, four-vertex-connected graph of girth four and Charney--Davis curvature one such that no four-vertex-connected ordinary subgraph admits a legal system. This resolves the $4$-connected problem stated by Jankiewicz, Norin, and Wise~\cite[Problem~5.2]{JNW21}. The proof combines an induced-$4$-cycle amalgam restriction, which preserves the obstruction to legal systems, with a $4$-regular host construction whose connectivity and curvature meet the required global conditions. In particular, the example shows that nonnegative Charney--Davis curvature together with $4$-connectivity does not force the existence of a legal system, even after passing to $4$-connected ordinary subgraphs.

\ifdefined\JNWCONTENTONLY
\else
\end{document}
\fi

\end{document}